\documentclass[11pt,a4paper,english]{amsart}
\usepackage[utf8]{inputenc}
\usepackage[english]{babel} 

\usepackage{amsmath} 
\usepackage{amsthm} 
\usepackage{graphicx} 
\usepackage{xcolor} 
\usepackage{amsfonts}
\usepackage[biblabel]{cite}
\usepackage{bm} 
\usepackage[hidelinks]{hyperref} 
\usepackage{amssymb} 
\usepackage{tikz-cd} 
\usepackage{amscd}
\usepackage{etoolbox}
\patchcmd{\thmhead}{(#3)}{#3}{}{}
\usepackage{enumitem}
\usepackage{breqn} 
\usepackage{faktor} 
\usepackage{xfrac} 

\DeclareMathOperator{\supp}{supp} 

\DeclareMathOperator{\RM}{RM}
\DeclareMathOperator{\GRM}{GRM}

\DeclareMathOperator{\wt}{wt}

\DeclareMathOperator{\RS}{RS}

\newcommand{\F}{{\mathbb{F}}}
\newcommand{\fq}{\mathbb{F}_q}

\newcommand{\B}{{\mathcal{B}}}

\usepackage{mathtools} 
\DeclarePairedDelimiter\abs{\lvert}{\rvert}%
\DeclarePairedDelimiter\norm{\lVert}{\rVert}%

\makeatletter
\let\oldabs\abs
\def\abs{\@ifstar{\oldabs}{\oldabs*}}
\let\oldnorm\norm
\def\norm{\@ifstar{\oldnorm}{\oldnorm*}}
\makeatother

\newtheorem{thm}{Theorem}[section]
\newtheorem{prop}[thm]{Proposition}
\newtheorem{cor}[thm]{Corollary}
\newtheorem{lem}[thm]{Lemma}
\theoremstyle{definition}
\newtheorem{defn}[thm]{Definition} 

\newtheorem{rem}[thm]{Remark} 
 
\newtheorem{ex}[thm]{Example}

\title[About the generalized Hamming weights of matrix-product codes]{About the generalized Hamming weights of matrix-product codes}
\author{Rodrigo San-José}
\curraddr{
\texttt{Rodrigo San-José:} IMUVA-Mathematics Research Institute, Universidad de Valladolid, 47011 Valladolid (Spain).
}
\email{rodrigo.san-jose@uva.es}

\thanks{This work was supported in part by the following grants: Grant TED2021-130358B-I00 funded by MICIU/AEI/10.13039/501100011033 and by the ``European Union NextGenerationEU/PRTR'', Grant PID2022-138906NB-C21 funded by MICIU/AEI/10.13039/501100011033 and by ERDF/EU, and FPU20/01311 funded by the Spanish Ministry of Universities.}



\begin{document}

\maketitle

\begin{abstract}
We derive a general lower bound for the generalized Hamming weights of nested matrix-product codes, with a particular emphasis on the cases with two and three constituent codes. We also provide an upper bound which is reminiscent of the bounds used for the minimum distance of matrix-product codes. When the constituent codes are two Reed-Solomon codes, we obtain an explicit formula for the generalized Hamming weights of the resulting matrix-product code. We also deal with the non-nested case for the case of two constituent codes.
\end{abstract}

\section{Introduction}

The generalized Hamming weights (GHWs) of a linear code were introduced by Wei in \cite{weiGHW}, and they are a generalization of the minimum distance. Indeed, the GHWs of a code are obtained as the minimum of the cardinalities of the supports of all its subcodes of a fixed dimension $r$, e.g., for $r=1$ one obtains the minimum distance. They give finer information about the code, and, in terms of applications, they characterize its performance on the wire-tap channel of type II and as a $t$-resilient function \cite{weiGHW}, they have applications to list decoding \cite{guruswammiGHWlistdecoding,guruswammiGHWlistdecodingTensorInterleaved}, their relative version has applications to secret sharing \cite{matsumotoRGHW}, and the rank-metric version has applications to network coding \cite{umbertoGHWandGRW,existenceGRW,matsumotoRGRW}. This has motivated the study of GHWs in general \cite{hellesethGHWLinearcodes}, as well as the computation of these parameters for well known families of codes, such as cyclic codes\cite{janwaGHWCyclic,yangGHWCyclic,fengGHWsCyclic} (also see \cite{hellesethGHWcyclic}), Reed-Muller codes \cite{pellikaanGHWRM}, Cartesian codes \cite{beelenGHWcartesian}, hyperbolic codes \cite{eduardoGHWHyperbolic}, and algebraic geometry codes \cite{munueraGHWhermitica,munueraGHWGoppa,sanjoseGHWNT}, among others. Nevertheless, the computation of the GHWs of a code is, in general, a difficult problem, and they are still unknown for many families of codes. 

Matrix-product codes (MPCs) were introduced by Blackmore and Norton in \cite{blackmoreMPC}. They have received a lot of attention since then \cite{aschMPC,fanMPC,liuMPChomogeneous1,liuMPChomogeneous2}, and they have found applications in many different contexts \cite{LuoSymbolpairMPC,galindoQuantumMPC,galindoMPCLRC,luoMPCLRC}. This technique utilizes an $s\times h$ matrix $A$ and $s$ linear codes $C_1,\dots,C_s$ of length $n$, and provides a new code of length $nh$ (see Definition \ref{d:mpc}). From the properties of the constituent codes, one can derive properties of the corresponding MPC. Most notably, one can obtain a lower bound for the minimum distance of the MPC from the minimum distance of the constituent codes \cite{blackmoreMPC,ferruhMPC}, but one can also derive self-orthogonality properties for some matrices \cite{galindoQuantumMPC,jitmanSelforthogonalMPCEU,jitmanSelforthogonalMPCHerm} or decoding algorithms \cite{decodingMPC,hernandoDecodingMPC2,hernandoListDecodingMPC}.

The aim of this work is to study the GHWs of an MPC in terms of those of its constituent codes. By doing this, one can consider families of codes with known GHWs, and derive different codes with bounded GHWs using the MPC construction. This allows us to substantially expand the families of codes for which we have bounds for their GHWs. This work can also be seen as a generalization of the bounds given for the minimum distance in \cite{blackmoreMPC,ferruhMPC}. In Section \ref{S:2codigos}, we focus on the case of $2\times 2$ matrices, without requiring the constituent codes to be nested. In Theorem \ref{T:doscodigos}, we give a lower bound for the GHWs of the corresponding MPC in terms of the GHWs of the constituent codes, and their sum and intersection. For the minimum distance of the code, this provides a refinement of the usual bounds for the $(u,u+v)$ and $(u+v,u-v)$ constructions (see \cite[Thm. 2.1.32 \& Prop. 2.1.39]{pellikaanlibro}), which is showcased in Example \ref{ex:uuv}. In Section \ref{S:generalbound}, by requiring the constituent codes to be nested, we generalize the techniques from Section \ref{S:2codigos} to obtain a lower bound for the GHWs of an MPC for an arbitrary non-singular by columns (NSC) matrix, and, in Subsections \ref{ss:h2} and \ref{ss:h3}, we describe it explicitly for the cases of two (Corollary \ref{C:h2}) and three (Theorem \ref{T:h3}) constituent codes. To complement these lower bounds, in Section \ref{S:upperbound} we provide an upper bound for the GHWs of MPCs, whose expression is reminiscent of the bounds obtained for the minimum distance in \cite{blackmoreMPC,ferruhMPC}. In Section \ref{S:examples}, we apply our results for specific families of codes. In particular, we show that our bounds are sharp when we consider two Reed-Solomon codes and a $2\times 2$ NSC matrix (Theorem \ref{T:GHWsRS}), therefore obtaining the weight hierarchy of these types of codes. We also test the bounds given in Corollary \ref{C:h2} and Theorem \ref{T:h3} for the case of two and three constituent Reed-Muller codes, and they give the true values of the GHWs in all the cases we have checked. 

\section{Preliminaries}

Let $\F_q$ be the finite field of $q$ elements, where $q$ is a power of a prime $p$. We start by defining MPCs as in \cite{blackmoreMPC}. 

\begin{defn}\label{d:mpc}
Let $C_1,\dots,C_s \subset \F_q^n$ be linear codes of length $n$, which we call \textit{constituent codes}, and let $A=(a_{ij})\in \F_q^{s\times h}$ be an $s\times h$ matrix, with $s\leq h$. Given $v_\ell\in C_\ell$, for $\ell=1,\dots,s$, we define
\begin{equation}\label{eq:codeword}
[v_1,\dots,v_s]\cdot A=\left(\sum_{\ell=1}^s a_{\ell 1}v_\ell,\dots, \sum_{\ell=1}^s a_{\ell h}v_\ell \right) \in \F_q^{nh}.
\end{equation}
Then the \textit{matrix-product code} $C$ associated to $A$ and $C_1,\dots,C_s$ is
$$
C=[C_1,\dots,C_s]\cdot A:= \left\{ [v_1,\dots,v_s]\cdot A : v_\ell\in C_\ell,\; \ell=1,\dots,s \right\} \subset \fq^{nh}.
$$
For each vector $c\in C$, we have a natural subdivision of the coordinates in $h$ blocks of length $n$, i.e., 
$$
c=(c_1,c_2,\dots,c_h),\; c_i\in \F_q^n, \; i=1\dots,h.
$$
\end{defn}

\begin{ex}
One can recover the usual $(u,u+v)$ construction (sometimes called Plotkin sum) of the codes $C_1$ and $C_2$ as an MPC code as follows:
$$
[C_1,C_2]\cdot \begin{pmatrix}
1&1\\
0&1
\end{pmatrix}=\{(v_1,v_1+v_2) : v_1\in C_1,\; v_2\in C_2\}.
$$    
\end{ex}

\begin{defn}\label{d:short}
We denote by $e_i$, $1\leq i \leq h$, the standard vectors of $\mathbb{Z}_2^h$, i.e., the vectors whose only nonzero entry is equal to 1 and it is in the $i$-th position. Let $y \in \mathbb{Z}_2^h$. Then we define
$$
C(y):=\{ c\in C : c_i=0 \text{ for each } i\in \supp(y) \}.
$$
In other words, $C(y)$ is similar to a shortening at the blocks given by $\supp(y)$, but without puncturing those coordinates. 
\end{defn}

Note that we are using subindices for vectors to express different things: to stress that a vector $v_\ell$ belongs to $C_\ell$, to denote the $i$-th block $c_i$ of a codeword $c\in C$, and to denote the standard vectors $e_i$ of $\mathbb{Z}_2^h$. We will use different letters ($v$, $c$ and $e$) and subindices ($i$ and $\ell$), which, together with the context, will help to clear any possible confusion. 

With respect to the parameters of MPCs, it is clear that the length is $nh$, and the dimension is $k=k_1+\dots+k_s$, where $k_\ell=\dim C_\ell$, $1\leq \ell \leq s$, if $A$ has full rank. In what follows, we always assume that $A$ has full rank. For the minimum distance, we have to introduce some notation. Let us denote by $R_\ell=(a_{\ell 1},\dots,a_{\ell h})$ the element of $\F_q^h$ given by the $\ell$-th row of $A$, for $1\leq \ell \leq s$. We denote by $\delta_\ell$ the minimum distance of the code $C_{R_\ell}$ generated by $\{ R_1,\dots,R_\ell\} $ in $\F_q^h$. In \cite{ferruhMPC} it is proven that
\begin{equation}\label{eq:cotadmin}
d_1(C)\geq \min \{d_1(C_1) \delta_1,\dots,d_1(C_s)\delta_s \} ,
\end{equation}
where $d_1(D)$ denotes the minimum distance the code $D$. Moreover, in \cite{decodingMPC}, the authors prove that the previous bound is sharp if $C_s\subset \cdots \subset C_1$.

When working with MPCs, it is usual to consider the following condition, introduced in \cite{blackmoreMPC}. 

\begin{defn}
Let $A$ be an $s\times h$ matrix, and let $A_t$ be the matrix formed by the first $t$ rows of $A$. For $1\leq j_i<\cdots < j_t\leq h$, we denote by $A(j_1,\dots,j_t)$ the $t\times t$ matrix consisting of the columns $j_1,\dots,j_t$ of $A_t$. A matrix $A$ is \textit{non-singular by columns} if $A(j_1,\dots,j_t)$ is non-singular for each $1\leq t \leq s$ and $1\leq j_1<\cdots < j_t\leq h$. In particular, an NSC matrix has full rank.
\end{defn}

\begin{ex}\label{ex:vandermonde}
Let $\F_q=\{\beta_1,\dots,\beta_q\}$. For $1\leq s\leq q$, the Vandermonde matrix
$$
V_m=\begin{pmatrix}
    1&\cdots&1\\
    \beta_1&\cdots & \beta_q\\
    \vdots &\ddots & \vdots \\
    \beta_1^{s-1}&\cdots & \beta_q^{s-1}
\end{pmatrix}
$$
is an NSC matrix. Also $V_M(j_1,\dots,j_h)$ is NSC for any $s\leq h\leq q$ and $1\leq j_1<\cdots<j_h\leq q$. 
\end{ex}

In \cite{blackmoreMPC} it is shown that, if $A$ is NSC, then the codes $C_{R_\ell}$ are MDS (i.e., $\delta_\ell=h-\ell+1$), for $1\leq \ell \leq s$. This implies that the bound (\ref{eq:cotadmin}) becomes
\begin{equation}\label{eq:cotadminNSC}
d_1(C)\geq \min \{h d_1(C_1) ,(h-1)d_1(C_2),\dots,(h-s+1)d_1(C_s) \} 
\end{equation}
for the case of an NSC matrix.

One of the goals of this work is to generalize the bounds (\ref{eq:cotadmin}) and (\ref{eq:cotadminNSC}) to the case of the GHWs of $C$, which we introduce now. Let $D\subset C$ be a subcode. The support of $D$, denoted by $\supp(D)$, is defined as
$$
\supp(D):=\{i: \exists\ u=(u_1,\dots,u_{nh})\in D,\; u_i\neq 0 \}.
$$
Note that, in this case, $u_i$ is just the $i$-th coordinate of $u$, not the $i$-th block of length $n$ of $u$. Let $1\leq r \leq \dim C$. The $r$-th \textit{generalized Hamming weight} of $C$, denoted by $d_r(C)$, is defined as
$$
d_r(C):=\min \{ \abs{\supp(D)}: D \text{ is a subcode of $C$ with } \dim D=r\},
$$
where $\abs{A}$ denotes the cardinality of a set $A$. Throughout the paper, we will denote $d_0(C)=0$. 

\begin{rem}\label{r:baseGHW}
Given a basis $B=\{b_1,\dots,b_k\}$ for a subcode $D$, we have that
$$
\supp(D)=\bigcup_{i=1}^k \supp(b_i).
$$
\end{rem}

The GHWs satisfy the following general properties for any linear code $C$, as shown in \cite{weiGHW}.

\begin{thm}[(Monotonicity)]\label{T:monotonia}
For an $[n,k]$ linear code $C$ with $k>0$ we have
$$
1\leq d_1(C)<d_2(C)<\cdots <d_k(C)\leq n.
$$
\end{thm}
\begin{cor}[(Generalized Singleton Bound)]\label{C:singletongeneralizada}
For an $[n,k]$ linear code $C$ we have
$$
d_r(C)\leq n-k+r, \; 1\leq r\leq k.
$$
\end{cor}

We say that a code $C$ is $t$-MDS if $d_t(C)=n-k+t$, for some $1\leq t \leq \dim C$. If a code is $t$-MDS for $t<\dim C$, it is also $(t+1)$-MDS by Theorem \ref{T:monotonia} and Corollary \ref{C:singletongeneralizada}. Thus, one usually studies what is the first $t$ such that $C$ is $t$-MDS.

\begin{rem}\label{r:ghwMDS}
For an MDS code $C$, by Theorem \ref{T:monotonia} and Corollary \ref{C:singletongeneralizada} we have
$$
d_r(C)=n-k+r,
$$
for all $1\leq r \leq k$.
\end{rem}

Going back to MPCs, the block structure that we have allows us to divide the support of the code as follows.

\begin{defn}
Let $C\subset \F_q^{nh}$. Then we define
$$
\supp_i(C):=\supp(C)\cap \{(i-1)\cdot n +1,\dots,i\cdot n \},\; 1\leq i \leq h.
$$
It is clear that 
$$
\supp(C)=\bigsqcup_{i=1}^h \supp_i(C).
$$
This implies
\begin{equation}\label{eq:sumasupport}
\abs{\supp(C)}=\sum_{i=1}^h \abs{\supp_i(C)}.
\end{equation}
\end{defn}

\section{A bound for the GHWs of the MPCs with $2\times 2$ matrices}\label{S:2codigos}
In this section, we give a lower bound for the GHWs of MPCs obtained with a $2\times 2$ matrix A, which we also assume to be NSC. If we denote
$$
A=\begin{pmatrix}
    a_{11} & a_{12} \\
    a_{21} & a_{22}
\end{pmatrix},
$$
since $A$ is NSC, we have $a_{1j}\neq 0$, $1\leq j\leq 2$. Moreover, we also cannot have $a_{21}=a_{22}=0$. Since exchanging the order of the columns of $A$ produces a permutation equivalent MPC code, we will assume that $a_{22}\neq 0$. Let $C_1,C_2\subset \fq^n$, and $1\leq r \leq \dim C_1+\dim C_2$. We denote $[r]:=\{0,1\dots,r\}$ and $[r]^2:=[r]\times [r]$. We consider the set
$$
    Y_2 := \left\{(\alpha_1,\alpha_2) \in [r]^2: \begin{array}{c}
        r-\dim(C_1+C_2) \leq  \alpha_1 \leq \dim C_2\\
        r-\dim(C_1+C_2) \leq  \alpha_2 \leq \dim (C_1\cap C_2)\\
        \alpha_1+\alpha_2\leq r
        \end{array}\right\}.
$$

We give now the main result of the section, bounding from below the GHWs of an MPC in terms of the GHWs of sums and intersections of the constituent codes.

\begin{thm}\label{T:doscodigos}
Let $C_1,C_2\subset \F_q^n$, and let $C=[C_1,C_2]\cdot A$, with $A$ as above. Let $1\leq r\leq \dim C=\dim C_1+\dim C_2$. Then
$$
d_r(C)\geq \min_{(\alpha_1,\alpha_2)\in Y_2}B_{\alpha_1,\alpha_2},
$$
where 
$$
B_{\alpha_1,\alpha_2}= \max\{d_{r-\alpha_1}(C_1+C_2),d_{\alpha_2}(C_1\cap C_2)\}+ \max\{d_{r-\alpha_2}(C_1+C_2),d_{\alpha_1}(C_2)\}. 
$$
\end{thm}
\begin{proof}
Let $D\subset C$ be a subcode with $\dim D=r$. We will associate a pair $(\alpha_1,\alpha_2)$ to $D$, and we will see that 
$$
\abs{\supp(D)}\geq B_{\alpha_1,\alpha_2}.
$$
We consider the following subcodes of $D$ (recall Definition \ref{d:short}):
$$
D_1=D(e_1),\; D_2=D(e_2), \text{ and } D_3=D/(D(e_1)+D(e_2)),
$$
where $D_3$ is regarded as a subcode of $D$ by fixing some set of representatives of the quotient vector space. It is clear that 
$$
D=D_1\oplus D_2\oplus D_3.
$$
If we denote $\alpha_1=\dim D_1$ and $\alpha_2=\dim D_2$, we have that $\dim D_3=r-\alpha_1-\alpha_2\geq 0$. Moreover, by (\ref{eq:sumasupport}), we have 
$$
\abs{\supp(D)}=\sum_{i=1}^2 \abs{\supp_i(D)}.
$$
Now we will bound $\abs{\supp_i(D)}$ from below, for $1\leq i \leq 2$. We start with $i=1$. Let $\B_1$, $\B_2$ and $\B_3$ be any fixed bases for $D_1$, $D_2$, and $D_3$, respectively. We consider the basis $\B=\B_1\cup \B_2\cup \B_3$ for $D$. We can use Remark \ref{r:baseGHW}, and notice that
$$
\supp_1(D_1)=\bigcup_{b\in \B_1}\supp_i(b)=\emptyset. 
$$
Therefore, $\supp_1(D)=\supp_1(D_2 \oplus D_3)$. Now we have two ways to bound $\abs{\supp_1(D_2\oplus D_3)}$:
\begin{enumerate}
    \item[(a)] Let $\B'$ be the set 
    $$
    \B':=\{c_1: c=(c_1,c_2)\in \B_2\cup \B_3\},
    $$
    that is, the set formed by the first block of the vectors in $\B_2\cup \B_3$, which has size $r-\alpha_1$. From the definition of MPCs (see (\ref{eq:codeword})), $\B'\subset C_1+C_2$. Moreover, $\B'$ is a linearly independent set because, otherwise, we would have a linear combination of vectors of $\B_2\cup \B_3$ in $D_1$, a contradiction. Thus,
    $$
    \abs{\supp_1(D)}=\abs{\supp_1(D_2\oplus D_3)}\geq d_{r-\alpha_1}(C_1+C_2).
    $$
    \item[(b)] We consider the set
    $$
    \B''=\{c_1: c=(c_1,c_2)\in \B_2\}.
    $$
    As the vectors of $\B_2$ are linearly independent and they have $c_2=0$, the vectors in $\B''$ are linearly independent. Let $c_1\in \B''$. Then
    $$
    (c_1,0)=[v_1,v_2]\cdot A=(a_{11}v_1+a_{21}v_2,a_{12}v_1+a_{22}v_2),
    $$    
    for some $v_1\in C_1$, $v_2\in C_2$. Hence,
    $$
    0=a_{12}v_1+a_{22}v_2 \implies v_1=(-a_{22}/a_{12}) v_2,
    $$
    since $a_{12}\neq 0$. We are assuming $a_{22}\neq 0$, which implies $v_1,v_2\in C_1\cap C_2$. Therefore, $c_1=a_{11}v_1+a_{21}v_2\in C_1\cap C_2$ and $\B''\subset C_1\cap C_2$. We have obtained
    $$
    \abs{\supp_1(D)}=\abs{\supp_1(D_2\oplus D_3)}\geq d_{\alpha_2}(C_1\cap C_2).
    $$
\end{enumerate}
Using both bounds, we get
$$
\abs{\supp_1(D)} \geq \max \{ d_{r-\alpha_1}(C_1+C_2),d_{\alpha_2}(C_1\cap C_2) \}.
$$
An analogous argument applies to $\supp_2(D)$, taking into account that $a_{21}$ can be zero. This means that in (b) we can only argue that $v_1,v_2\in C_2$. We obtain the bound
$$
\abs{\supp_2(D)} \geq \max \{ d_{r-\alpha_2}(C_1+C_2),d_{\alpha_1}(C_2) \}.
$$
Thus, 
$$
\abs{\supp(D)}=\abs{\supp_1(D)}+\abs{\supp_2(D)}\geq  B_{\alpha_1,\alpha_2}.
$$
For any subcode $D$, from the arguments in (a) and (b) we deduce that the parameters $\alpha_1=\dim D(e_1)$ and $\alpha_2=\dim D(e_2)$ satisfy $(\alpha_1,\alpha_2)\in Y$, which concludes the proof. 
\end{proof}

We have given the bound in the most general form. However, depending on whether $a_{21}$ is zero or not, it is possible to improve the bound from the previous result, as we show next. If $a_{21}\neq 0$, consider the set
$$
Y_2' := \left\{(\alpha_1,\alpha_2) \in [r]^2: \begin{array}{c}
        r-\dim(C_1+C_2)\leq  \alpha_i \leq \dim C_1\cap C_2, \; i=1,2  \\
        \alpha_1+\alpha_2\leq r
        \end{array}\right\}.
$$

\begin{cor}\label{C:differentboundsh2no0}
With the notation as before, if $a_{21}\neq 0$, then
$$
d_r(C)\geq \min_{(\alpha_1,\alpha_2)\in Y_2'}B_{\alpha_1,\alpha_2},
$$
where 
$$
B_{\alpha_1,\alpha_2}= \max\{d_{r-\alpha_1}(C_1+C_2),d_{\alpha_2}(C_1\cap C_2)\}+ \max\{d_{r-\alpha_2}(C_1+C_2),d_{\alpha_1}(C_1\cap C_2)\}. 
$$
\end{cor}
\begin{proof}
Following the proof of Theorem \ref{T:doscodigos}, if $a_{21}\neq 0$, then in (b) we have $v_1,v_2\in C_1\cap C_2$ for both blocks $i=1,2$.
\end{proof}

In the case of $a_{21}=0$, we consider instead the set
$$
    Y_2'' := \left\{(\alpha_1,\alpha_2) \in [r]^2: \begin{array}{c}
        r-\dim(C_1) \leq  \alpha_1 \leq \dim C_2\\
        r-\dim(C_1+C_2) \leq  \alpha_2 \leq  \dim C_1\cap C_2\\
        \alpha_1+\alpha_2\leq r
        \end{array}\right\}.
$$

\begin{cor}\label{C:differentboundsh20}
With the notation as before, if $a_{21}=0$, then
$$
d_r(C)\geq \min_{(\alpha_1,\alpha_2)\in Y_2''}B_{\alpha_1,\alpha_2},
$$
where 
$$
B_{\alpha_1,\alpha_2}= \max\{d_{r-\alpha_1}(C_1),d_{\alpha_2}(C_1\cap C_2)\}+ \max\{d_{r-\alpha_2}(C_1+C_2),d_{\alpha_1}(C_2)\}. 
$$
\end{cor}
\begin{proof}
We follow the proof of Theorem \ref{T:doscodigos}. If $a_{21}=0$, then for any $c\in C$, we have $c_1\in C_1$, improving the bound obtained in (a) for the first block.
\end{proof}

\begin{rem}
The ideas in this section are a generalization of the arguments from \cite{sanjoseRecursivePRM}, where the author computes a particular generator matrix for any subcode of a projective Reed-Muller code that is given by two parameters, $\alpha$ and $\gamma$. Those parameters play the role of $r-\alpha_2$ and $\alpha_1$, respectively, in this section. 
\end{rem}

Note that, if $C_2\subset C_1$, then all the bounds given in this section coincide. However, as we show in the next example, if we do not have this nested condition, then we can obtain different bounds in Corollaries \ref{C:differentboundsh2no0} and \ref{C:differentboundsh20}. Moreover, in the next example we also show that, if the codes are not nested, our bounds can refine the usual bounds for the minimum distance of the $(u,u+v)$ and $(u+v,u-v)$ constructions by considering $d_1(C_1+C_2)$ and $d_1(C_1\cap C_2)$.

\begin{ex}\label{ex:uuv}
Let $q=3$, and consider 
$$
G_1=\begin{pmatrix}
    0 & 1 & 0 & 0 & 0 & 1 & 1 & 0 \\
-1 & 1 & 0 & 1 & -1 & 1 & 0 & 1 \\
-1 & 1 & -1 & 1 & 1 & 1 & 1 & 0
\end{pmatrix}, \; G_2=\begin{pmatrix}
    -1 & 0 & 1 & 1 & -1 & 1 & -1 & 0 \\
1 & 1 & 0 & 1 & -1 & -1 & -1 & -1
\end{pmatrix}.
$$
Let $C_1$ and $C_2$ be the linear codes whose generator matrices are $G_1$ and $G_2$. Then, one can check that $C_1\cap C_2=\{0\}$, and the GHWs of $C_1$, $C_2$ and $C_1+C_2$ are given in Table \ref{table:GHWSej1}.

\begin{table}[ht]
\caption{GHWs of $C_1$, $C_2$ and $C_1+C_2$}
\centering
\label{table:GHWSej1}
\begin{tabular}{c|ccccccccccccccc}
GHWs$\backslash r$ & 1&2&3&4&5 \\
\hline
$d_r(C_1)$& 3&6&8&-&- \\
$d_r(C_2)$& 5&8&-&-&-\\
$d_r(C_1+C_2)$& 3&5&6&7&8\\
\end{tabular}
\end{table}

Now we define the matrices
$$
A_1:=\begin{pmatrix}
    1&1\\
    0&1
\end{pmatrix},\; A_2:=\begin{pmatrix}
    1&1\\
    1&-1
\end{pmatrix},
$$
which correspond to the $(u,u+v)$ and $(u+v,u-v)$ constructions, respectively. Let $D_1=[C_1,C_2]\cdot A_1$, $D_2=[C_1,C_2]\cdot A_2$. The usual bounds for the minimum distance of $D_1$ and $D_2$ would give $\min \{ 2d_1(C_1),d_1(C_2) \}=5$ (see \cite[Thm. 2.1.32 \& Prop. 2.1.39]{pellikaanlibro}). However, our bounds from Corollaries \ref{C:differentboundsh2no0} and \ref{C:differentboundsh20} give the values from Table \ref{table:u+vu-v}.

\begin{table}[ht]
\caption{Lower bounds from Corollaries \ref{C:differentboundsh2no0} and \ref{C:differentboundsh2no0}}
\centering
\label{table:u+vu-v}
\begin{tabular}{c|ccccccccccccccc}
Bound$\backslash r$ & 1&2&3&4&5 \\
\hline
Lower bound for $D_1$& 5&8&11&14&16\\
Lower bound for $D_2$& 6&10&12&14&16\\
\end{tabular}
\end{table}

Note that the bound for $d_1(D_2)$ has been improved to $6$. Also, notice that the bounds obtained from Corollaries \ref{C:differentboundsh2no0} and \ref{C:differentboundsh20} are different in this case for $A_1$ and $A_2$. This is noteworthy since, as we said before, the usual bounds for the minimum distance of the $(u,u+v)$ construction and the $(u+v,u-v)$ construction are the same (assuming $d_1(C_1)\leq d_1(C_2)$). The true values of the GHWs are given in Table \ref{table:trueu+vu-v}, showing that our bounds are sharp in this case, except in the case $r=4$ for $D_2$. 

\begin{table}[ht]
\caption{GHWs of $D_1$ and $D_2$}
\centering
\label{table:trueu+vu-v}
\begin{tabular}{c|ccccccccccccccc}
GHWs$\backslash r$ & 1&2&3&4&5 \\
\hline
$d_r(D_1)$& 5&8&11&14&16\\
$d_r(D_2)$& 6&10&12&15&16\\
\end{tabular}
\end{table}
In this case, since $C_1\cap C_2=\{0\}$, the lower bounds from Corollaries \ref{C:differentboundsh2no0} and \ref{C:differentboundsh20} are particularly easy to compute. Indeed, if $a_{21}\neq 0$ (the case of $A_2$), we have $Y_2'=\{(0,0)\}$. Thus, the bound from Corollary \ref{C:differentboundsh2no0} is just
$$
d_r(D_2)\geq B_{0,0}=2d_r(C_1+C_2). 
$$
For the case $a_{21}=0$, we obtain $Y_2''=\{(\alpha_1,0)\in [r]^2: r-3\leq \alpha_1\leq 2 \}$, and the bound from Corollary \ref{C:differentboundsh20} is
$$
d_r(D_1)\geq \min_{(\alpha_1,0)\in Y_2''} B_{\alpha_1,0}=\min_{(\alpha_1,0)\in Y_2''} \{d_{r-\alpha_1}(C_1)+\max \{ d_r(C_1+C_2),d_{\alpha_1}(C_2)\}\}. 
$$
For example, for $r=3$, we have $Y_2''=\{(0,0),(1,0),(2,0)\}$, and
$$
d_3(D_1)\geq \min \{ 8+\max\{3,0\},6+\max\{6,5 \},3+\max\{ 6,8\}\}=11.
$$
\end{ex}

\section{A bound for the GHWs of nested MPCs with NSC matrices}\label{S:generalbound}
In this section, we will show how to obtain a lower bound for the GHWs of MPCs with $s$ constituent codes. We will assume that the codes are nested, i.e., $C_s\subset  \cdots \subset C_1 \subset \F_q^n$. We consider $A$ an $s\times h$ NSC matrix over $\F_q$ with $s\leq h$. By \cite[Prop. 3.3]{blackmoreMPC}, this implies that $h\leq q$. Let $C=[C_1,\dots,C_s]\cdot A$. From \cite[Lem. 6]{LuoSymbolpairMPC} we have the following result.

\begin{lem}\label{l:mpcsceros}
Let $C_s\subset \cdots \subset C_1 \subset \F_q^n$ and $A$ an $s\times h$ NSC matrix over $\F_q$. Let $C=[C_1,\dots,C_s]\cdot A$ and $c\in C$. We consider the $h$ blocks of length $n$ of $c$, that is, $c=(c_1,\dots,c_h)$. Let $0\leq \ell \leq s-1$. If there are exactly $\ell$ zero vectors among the blocks $c_1,\dots,c_h$, then $c_j\in C_{\ell +1}$, for every $1\leq j\leq h$. If the number of zero vectors among $c_1,\dots,c_h$ is greater than $s-1$, then $c=0$. 
\end{lem}

Let $1\leq i \leq h$, $1\leq r \leq \dim C=\sum_{\ell=1}^s \dim (C_\ell)$, and $D\subset C$ a subcode of $C$ with $\dim D=r$. For $0\leq j \leq s-1$, we define the vector space
\begin{equation}\label{eq:quotientspace}
D^i_j:=\left. \left(D(e_i) +\sum_{y\in \mathbb{Z}_2^h,\; \wt(y)=j}D(y)\right)\middle/  D(e_i).\right. 
\end{equation}
One way to think about $D_j^i$ is that these are the codewords $c\in D$ with $c_i\neq 0$ (because we take the quotient by $D(e_i)$), and which can be generated by codewords of $D$ with at least $j$ zero blocks. Note that, since $D(y)\subset D(e_i)$ if $y_i=1$, we have 
$$
D(e_i) +\sum_{y\in \mathbb{Z}_2^h,\; \wt(y)=j}D(y)=D(e_i)+\sum_{y\in \mathbb{Z}_2^h,\; \wt(y)=j,\; y_i=0}D(y).
$$

\begin{ex}\label{ex:Dquotient}
For $s=h=2$ and $1\leq i \leq 2$ we have
$$
D_0^i=(D(e_i)+D)/D(e_i)= D/D(e_i),\; D_1^i=(D(e_1)+D(e_2))/D(e_i)\cong D(e_{i+1}),
$$
where we understand the subindex $i+1$ cyclically mod 2, i.e., $2+1\equiv 1$. Note that the vector spaces $D_1^i$ already appeared in the proof of Theorem \ref{T:doscodigos}.
\end{ex}

We can consider a basis for this last vector space where every vector is either in some $D(y)$, with $\wt(y)=j$, $y_i=0$, or in $D(e_i)$. The classes of these vectors in (\ref{eq:quotientspace}) form a generating set, from which we can extract a basis $\B^i_j$ (regarded in $\F_q^{hn}$ by fixing some representatives) where every vector is in some $D(y)$, with $\wt(y)=j$, and is not contained in $D(e_i)$. That is, each vector of $\B^i_j$ has at least $j$ zero blocks, and its $i$-th block is nonzero. We now define
$$
\B_{j,i}^i:=\{c_i: c\in \B^i_j\},
$$
which is the set given by the $i$-th blocks of the vectors in $\B^i_j$. 

\begin{ex}
Following the setting of Example \ref{ex:Dquotient}, we have
$$
\abs{\B_0^i}=\dim D_0^i=r-\dim D(e_i),\; \abs{\B_1^i}=\dim D_1^i=\dim D(e_{i+1}).
$$
\end{ex}

Now we can use the sets we have just defined to obtain a bound for $\abs{\supp_i(D)}$.

\begin{lem}\label{l:suppiD}
We have that 
$$
\abs{\supp_i(D)} \geq d_{\abs{\B_j^i}}(C_{j+1}). 
$$
\end{lem}
\begin{proof}
We claim that $\B_{j,i}^i$ is a linearly independent set. Indeed, if we assume it is linearly dependent, this would give a linear combination of vectors of $\B_j^i$ in $D(e_i)$, a contradiction, since the classes of the vectors of $\B_j^i$ are linearly independent in $D_j^i$ (see (\ref{eq:quotientspace})). By Lemma \ref{l:mpcsceros}, we have $\B_{j,i}^i\subset C_{j+1}$, and
$$
\abs{\supp_i(D)}\geq \abs{\bigcup_{b\in \B_{j,i}^i}\supp_i(b)} \geq d_{\abs{\B_{j,i}^i}}(C_{j+1})=d_{\abs{\B_j^i}}(C_{j+1}). 
$$

\end{proof}

From this lemma we can obtain a general result bounding the GHWs of an MPC. Note that $\B_j^i$ depends on the subcode $D$, and we could write $\B_j^i(D)$ to make this explicit, but we avoid doing this for ease of notation.

\begin{prop}\label{p:generallowerbound}
Let $C_s\subset  \cdots \subset C_1 \subset \F_q^n$ be linear codes, $A$ an $s\times h$ NSC matrix over $\F_q$ with $s\leq h$, and $C=[C_1,\dots,C_s]\cdot A$. For $1\leq r \leq \dim C=\sum_{\ell=1}^s\dim C_\ell$, we have
$$
d_r(C)\geq \min_{D\subset C, \; \dim D=r}\left(\sum_{i=1}^h \max\{ d_{\abs{\B_j^i}}(C_{j+1}),\; 0\leq j \leq s-1 \}\right).
$$
\end{prop}
\begin{proof}
Let $D\subset C$ be a subcode with $\dim D=r$. Using Lemma \ref{l:suppiD} for every block $i$, $1\leq i \leq h$, and taking into account (\ref{eq:sumasupport}), we obtain the bound 
$$
\abs{\supp(D)}\geq \sum_{i=1}^h \max\{ d_{\abs{\B_j^i}}(C_{j+1}),\; 0\leq j \leq s-1 \}.
$$
The result follows from the definition of GHWs. 
\end{proof}

\begin{rem}\label{r:generalizadmin}
For the case $r=1$, this bound generalizes the bound from (\ref{eq:cotadminNSC}). Indeed, let $D\subset C$ with $\dim D=1$, and consider $i,j$ such that $\abs{\B_{j}^i}=1$ (since $r=1$, $\abs{\B_{j}^i}$ is either 0 or 1, and if all of them are 0, this would correspond to the subcode $D=\{0\}$). This means that $D$ is generated by a vector $c$ with at least $j$ zero blocks, and with a nonzero $i$-th block. Let 
$$
j':=\abs{\{k: c_k= 0\}},
$$
that is, the number of zero blocks of $c$. Then $\abs{\B_{j'}^i}=1$ since we can assume $\B_{j'}^i=\{c\}$. It follows from the definitions that, in this case, we have
$$
\abs{\B_{k}^i}=1 \iff k\leq j',\; c_i\neq 0,
$$
and, thus, $\abs{\B_{k}^i}=0$ otherwise. Then, for any $i$ such that $c_i\neq 0$, we have
$$
\max\{ d_{\abs{\B_j^i}}(C_{j+1}),\; 0\leq j \leq s-1 \}=\max\{d_1(C_1),\dots,d_{1}(C_{j'+1})\}=d_{1}(C_{j'+1}).
$$
Since $c$ has exactly $h-j'$ nonzero blocks, we obtain
$$
\sum_{i=1}^h \max\{ d_{\abs{\B_j^i}}(C_{j+1}),\; 0\leq j \leq s-1 \}=(h-j')d_1(C_{j'+1}),
$$
which shows that the bound from Proposition \ref{p:generallowerbound} simplifies to (\ref{eq:cotadminNSC}) in this case. 
\end{rem}

The advantage of using Proposition \ref{p:generallowerbound} to compute the GHWs of $C$ instead of the definition is that, even though the bound from Proposition \ref{p:generallowerbound} requires to compute a minimum over all the subcodes $D\subset C$ with $\dim D=r$, the values we are minimizing only depend on $\abs{\B_j^i}$, e.g., see Remark \ref{r:generalizadmin}. Now assume we have a set $Y_s$ and a family of bounds $\{B_v\}_{v\in Y_s}$ such that for any subcode $D\subset C$ with $\dim D=r$, we have
$$
\sum_{i=1}^h \max\{ d_{\abs{\B_j^i}}(C_{j+1}),\; 0\leq j \leq s-1 \} =B_v,
$$
for some $v\in Y_s$. From Proposition \ref{p:generallowerbound} we obtain
\begin{equation}\label{eq:cotageneralY}
d_r(C)\geq \min_{D\subset C, \; \dim D=r}\left(\sum_{i=1}^h \max\{ d_{\abs{\B_j^i}}(C_{j+1}),\; 0\leq j \leq s-1 \}\right) \geq \min_{v\in Y_s}B_v.
\end{equation}

In the next subsections we show how to obtain a set $Y_s$ and a family of bounds $\{B_v\}_{v\in Y_s}$ for the cases of $s=2$ and $s=3$, which are the most used cases for applications.

\subsection{The case $h=2$}\label{ss:h2}
For the case $s=h=2$ we can recover what we obtained in Section \ref{S:2codigos} for the nested case. We recall that if $C_2\subset C_1$, we have $Y_2=Y'_2=Y''_2$ (using the notation from Section \ref{S:2codigos}), and
$$
Y_2 = \left\{(\alpha_1,\alpha_2)\in [r]^2: \begin{array}{c}
        r-\dim C_1 \leq  \alpha_i \leq\dim C_2, \; 1\leq i\leq 2 \\
        \alpha_1+\alpha_2\leq r
        \end{array}
        \right\}.
$$

\begin{cor}\label{C:h2}
Let $C_2\subset C_1\subset \fq^n$, $C=[C_1,C_2]\cdot A$, for some $2\times 2$ NSC matrix $A$. Consider $1\leq r \leq \dim C_1+\dim C_2$, and
$$
B_{\alpha_1,\alpha_2}= \max\{d_{r-\alpha_1}(C_1),d_{\alpha_2}(C_2)\}+ \max\{d_{r-\alpha_2}(C_1),d_{\alpha_1}(C_2)\},
$$
for $(\alpha_1,\alpha_2)\in Y_2$. Then 
$$
d_r(C)\geq \min_{(\alpha_1,\alpha_2)\in Y_2}B_{\alpha_1,\alpha_2}.
$$
\end{cor}
\begin{proof}
Let $D\subset C$ with $\dim D=r$. Let $\alpha_i=\dim D(e_i)$, $1\leq i \leq 2$, and note that $\abs{B_0^i}=r-\alpha_{i+1}$ (we consider $i+1\bmod 2$ for the subindex, with representatives $\{1,2\}$), and $\abs{B_1^i}=\alpha_i$. The first set of conditions about $\alpha_i$, $1\leq i \leq 2$, follow from the fact that $\B_{j,i}^i\subset C_{j+1}$ and $\abs{\B_j^i}=\abs{\B_{j,i}^i}$, for $j=0,1$. The condition $\alpha_1+\alpha_2\leq r$ arises from the fact that $D(e_1)+D(e_2)\subset D$, and $D(e_1)\cap D(e_2)=\{0\}$. Therefore, by Proposition \ref{p:generallowerbound} and (\ref{eq:cotageneralY}), we obtain the result. 
\end{proof}

\subsection{The case $h=3$}\label{ss:h3}
We now apply our techniques to the case $s=h=3$. 
Throughout this section, when a subindex is greater than $3$, we consider its reduction modulo $3$, with representatives $\{1,2,3\}$. For instance, for $i=2$, we have $e_{i+1}+e_{i+2}=e_3+e_1$. We denote $[r]^{3,3,1}:=[r]^{3}\times[r]^{3}\times [r] $, and let
$$
{
    Y_3 := \left\{(\alpha,\gamma,\beta)\in [r]^{3,3,1} : \begin{array}{c}
        \gamma_i\leq \dim C_3,\;1\leq i\leq 3\\
        \max\{r-\dim C_1,\gamma_{i+1}+\gamma_{i+2}\} \leq  \alpha_i ,\; 1\leq i\leq 3\\
        \alpha_{i+1}+\alpha_{i+2}-\gamma_i\leq \beta,\; 1\leq i \leq 3 \\
        \beta \leq \min \left\{ \displaystyle \sum_{i=1}^3(\alpha_i-\gamma_i),\dim C_2+\min \{\alpha_i,1\leq i\leq 3\}\right\}\\
        \end{array} \hspace{-0.1cm}\right\}.
}
$$

\begin{thm}\label{T:h3}
Let $C_3\subset C_2\subset C_1\subset \fq^n$ and $C=[C_1,C_2,C_3]\cdot A$, for some $3\times 3$ NSC matrix $A$. Consider $1\leq r \leq \sum_{\ell=1}^3 \dim C_\ell$. For $(\alpha,\gamma,\beta)\in Y_3$, let
$$
B_{\alpha,\gamma,\beta}=\sum_{i=1}^3 \max \{ d_{r-\alpha_i}(C_1),d_{\beta-\alpha_i}(C_2),d_{\gamma_i}(C_3)\}.
$$
Then we have
$$
d_r(C)\geq \min_{(\alpha,\gamma,\beta)\in Y_3} B_{\alpha,\gamma,\beta}.
$$
\end{thm}
\begin{proof}
Let $D\subset C$ with $\dim D=r$. We consider $\alpha_i=\dim D(e_i)$, $\gamma_i=\dim D(e_{i+1}+e_{i+2})$ and $\beta=\dim(\sum_{j=1}^3D(e_j))$, for $1\leq i\leq 3$. We claim
\begin{equation}\label{eq:bijh3}
\abs{\B^i_j}=\begin{cases}
    \dim D/D(e_i) = r-\alpha_i &\text{ if $j=0$,} \\
    \dim (\sum_{k=1}^3 D(e_k))/D(e_i)=\beta-\alpha_i &\text{ if $j=1$,}\\
    \dim (D(e_i)+\sum_{k<\ell}D(e_k+e_{\ell}))/D(e_i)=\gamma_i &\text{ if $j=2$.}
    \end{cases}
\end{equation}
The cases $j=0$ and $j=1$ are straightforward. For $j=2$, we have
$$
D(e_i)+\sum_{k<\ell}D(e_k+e_{\ell})=D(e_i)+D(e_{i+1}+e_{i+2})
$$
since $D(e_i+e_j)\subset D(e_i)$, for any $j\neq i$. Taking into account that $ D(e_i)\cap D(e_{i+1}+e_{i+2})=D((1,1,1))=\{0\}$, we have
$$
\dim (D(e_i)+\sum_{k<\ell}D(e_k+e_{\ell}))/D(e_i)=\dim(D(e_i)+D(e_{i+1}+e_{i+2}))-\dim D(e_i)=\gamma_i. 
$$
Let $\alpha=(\alpha_1,\alpha_2,\alpha_3)$, and $\gamma=(\gamma_1,\gamma_2,\gamma_3)$. Now we check that $(\alpha,\gamma,\beta)\in Y_3$ (we want to use (\ref{eq:cotageneralY})). It is clear that $0\leq \gamma_i$, and, since $\gamma_i=\abs{\B_2^i}=\abs{\B_{2,i}^i}$ and $\B_{2,i}^i\subset C_3$, we have $\gamma_i\leq \dim C_3$, for $1\leq i \leq 3$. Similarly, we have $r-\alpha_i=\abs{B_0^i}$, which implies $r-\alpha_i\leq \dim C_1$, i.e., $r-\dim C_1\leq \alpha_i$, for $1\leq i\leq 3$. Now we note that
$$
D(e_i+e_{i+2})+D(e_i+e_{i+1})\subset D(e_i).
$$
Taking into account that $D(e_i+e_{i+2})\cap D(e_i+e_{i+1})=D((1,1,1))=\{0\}$, we deduce that $\gamma_{i+1}+\gamma_{i+2}\leq \alpha_i$, $1\leq i\leq 3$. Regarding the first condition for $\beta$ in $Y$, we note that
$$
\beta=\dim\left(\sum_{i=1}^3D(e_i)\right)\geq \dim (D(e_{k+1})+D(e_{k+2}))=\alpha_{k+1}+\alpha_{k+2}-\gamma_k,
$$
for $1\leq k\leq 3$. It is clear that $\beta \leq r$, and, since $\beta-\alpha_i=\abs{\B_1^i}=\abs{\B_{1,i}^i}$ and $\B_{1,i}^i\subset C_2$, we have $\beta-\alpha_i\leq \dim C_2$, $1\leq i\leq 3$. The last condition we need to prove is $\beta\leq \sum_{i=1}^3(\alpha_i-\gamma_i)$. Note that, using the formula for the dimension of the sum of vector spaces twice, we have
$$
\dim\left(\sum_{i=1}^3D(e_i)\right)=\sum_{i=1}^3\alpha_i-\gamma_k-\dim(D(e_k)\cap(D(e_{k+1})+D(e_{k+2}))),
$$
for any $1\leq k \leq 3$. Since $D(e_k+e_{k+1})+D(e_k+e_{k+2}) \subset D(e_k)\cap(D(e_{k+1})+D(e_{k+2}))$, we conclude
$$
\beta=\dim\left(\sum_{i=1}^3D(e_i)\right) \leq \sum_{i=1}^3\alpha_i- \gamma_k-(\gamma_{k+2}+\gamma_{k+1})=\sum_{i=1}^3(\alpha_i-\gamma_i).
$$
Thus, we have proved that $(\alpha,\gamma,\beta)\in Y_3$ and, if we note the expressions in (\ref{eq:bijh3}), we have also proved that
$$
\sum_{i=1}^3 \max\{ d_{\abs{\B_j^i}}(C_{j+1}),\; 0\leq j \leq 3-1 \}=B_{\alpha,\gamma,\beta},
$$
for some $(\alpha,\gamma,\beta)\in Y_3$. We obtain the result by (\ref{eq:cotageneralY}). 
\end{proof}

\begin{rem}
As we have seen in the proof of the previous result, we have incorporated some of the relations between the dimensions of $D(e_i)$, $D(e_{i+1}+e_{i+2})$ and $\sum_{k=1}^3D(e_k)$, for $1\leq i \leq 3$, using $\alpha_i$, $\gamma_i$ and $\beta$, respectively. In fact, many of the relations between these dimensions that one could expect can be derived from the ones included in the definition of $Y$. For example, we have
$$
\dim(D(e_i))+\dim(D(e_{i+1}+e_{i+2}))=\dim(D(e_i)+D(e_{i+1}+e_{i+2}))\leq \dim\left(\sum_{i=1}^3D(e_i)\right).
$$
This means that we should have $\alpha_i+\gamma_i\leq \beta$, for $1\leq i\leq 3$. This is a consequence of the conditions we gave for $Y$ because
$$
\beta \geq \alpha_{i+1}+\alpha_{i+2}-\gamma_i\geq \alpha_{i+1}+\gamma_{i+1},\; 1\leq i \leq 3,
$$
since we also impose the condition $\alpha_{i+2}\geq \gamma_{i}+\gamma_{i+1}$. 
\end{rem}

Theorem \ref{T:h3} can also be used to give a bound for the GHWs in the case $s=2$, $h=3$, as the next result shows. In this case, we denote $[r]^{3,1}=[r]^3\times [r]$. 

\begin{cor}\label{C:h3s2}
Let $C_2\subset C_1\subset \fq^n$, $C=[C_1,C_2]\cdot A$, for some $2\times 3$ NSC matrix $A$. Let 
$$
{
    Y_3' = \left\{(\alpha,\beta)\in [r]^{3,1} : \begin{array}{c}
        r-\dim C_1\leq  \alpha_i ,\; 1\leq i\leq 3\\
        \alpha_{i+1}+\alpha_{i+2}\leq \beta,\; 1\leq i \leq 3 \\
        \beta \leq \min \left\{ \displaystyle \sum_{i=1}^3\alpha_i,\dim C_2+\min \{\alpha_i,1\leq i\leq 3\}\right\}\\
        \end{array} \hspace{-0.1cm}\right\}.
}
$$
For $(\alpha,\beta)\in Y_3'$ we consider
$$
B_{\alpha,\beta}=\sum_{i=1}^3\max \{d_{r-\alpha_i}(C_1),d_{\beta-\alpha_i}(C_2)\}. 
$$
Then we have
$$
d_r(C)\geq \min_{(\alpha,\beta)\in Y_3'}B_{\alpha,\beta}.
$$
\end{cor}
\begin{proof}
This can be obtained directly from Theorem \ref{T:h3} by setting $C_3=\{0\}$.  
\end{proof}

\begin{ex}\label{ex:h3s2}
Let $q=4$ and $n=4$. In this example (and throughout the rest of the paper) we denote by $\RS(k)$ the Reed-Solomon code of length $n$ and dimension $k$. Note that, by Remark \ref{r:ghwMDS}, we know the GHWs of Reed-Solomon codes. Let $k_1=3$ and $k_2=1$. We will compute the bound from Corollary \ref{C:h3s2} for the code $C=[\RS(k_1),\RS(k_2)]\cdot A$ and $r=2$, where
$$
A=\begin{pmatrix}
    1&a&1\\
    1&1&0
\end{pmatrix},
$$
and where $a$ is a primitive element of $\F_4$. We start by computing $Y_3'$. First, we have $0\leq \alpha_i \leq r=2$, for $1\leq i \leq 3$.  For $\beta$, we have the conditions $\alpha_{i+1}+\alpha_{i+2}\leq \beta$, for $1\leq i \leq 3$, and $\beta \leq \min \{\sum_{i=1}^3\alpha_i,1+\min\{\alpha_i,1\leq i \leq 3\}\}$, besides the condition $\beta\leq r=2$. It is straightforward to check that $\{(0,0,0)\}\times \{0\}\in Y_3'$. If we consider $\alpha=(1,0,0)$, then, looking at the conditions for $\beta$, this implies $\beta=1$, and we have $\{(1,0,0)\}\times \{1\}\in Y_3'$. Similarly, we have $\{(0,1,0)\}\times \{1\},\{(0,0,1)\}\times \{1\}\in Y_3'$. Finally, if we take $\alpha=(1,1,1)$, this implies $\beta=2$ and $\{(1,1,1)\}\times \{ 2\}\in Y_3'$. In fact, one can check that these are all the elements of $Y_3'$. For example, if we have $\alpha=(1,1,0)$, then we must also have $\alpha_1+\alpha_2=2\leq \beta$, but $\beta \leq 1+\min\{\alpha_i,1\leq i \leq 3\}=1$, a contradiction. A similar reasoning applies to $\alpha=(1,0,1)$ or $\alpha=(0,1,1)$, and also for the cases where $\alpha_i=2$ for some $1\leq i \leq 3$. 

Therefore, we have
$$
Y_3'=\{ \{(0,0,0)\}\times \{0\},\{(1,0,0),(0,1,0),(0,0,1) \} \times \{1\},\{(1,1,1)\}\times \{ 2\}\}. 
$$
Now we compute $B_{\alpha,\beta}$, for each $(\alpha,\beta)\in Y_3'$:
$$
\begin{aligned}
    &B_{(0,0,0),0}=3d_2(\RS(k_1))=3(n-k_1+2)=9,\\
    &B_{(1,0,0),1}=B_{(0,1,0),1}=B_{(0,0,1),1}=d_1(\RS(k_1))+2\max \{d_2(\RS(k_1)),d_1(\RS(k_2)) \}=10, \\
    &B_{(1,1,1),1}=3d_1(\RS(k_2))=3(n-k_2+1)=12.
\end{aligned}
$$
Hence, we obtain
$$
d_2(C)\geq \min_{(\alpha,\beta)\in Y_3'}B_{\alpha,\beta}=9.
$$
It can be checked with a computer that this is the true value of $d_2(C)$. 
\end{ex}

\section{An upper bound for the GHWs}\label{S:upperbound}
In this section we give an upper bound for the GHWs of MPCs, complementing the previous section, as this will allow us to ensure that our bound is sharp when both bounds coincide. For this result, we do not require $A$ to be NSC. We recall that $R_\ell=(a_{\ell 1},\dots,a_{\ell h})$ is the $\ell$-th row of $A$, for $1\leq \ell \leq s$; $\delta_\ell$ is the minimum distance of the code $C_{R_\ell}$ generated by $\{ R_1,\dots,R_\ell\} $; and $A_{\ell}$ is the matrix formed by the first $\ell$ rows of $A$. The proof of the following result is a generalization of the proof in \cite[Thm. 1]{decodingMPC} for the minimum distance. 

\begin{prop}\label{P:upperbound}
Let $C_s\subset \cdots \subset C_1\subset \fq^n$, and $C=[C_1,\dots,C_s]\cdot A$, where $A\subset \F_q^{s\times h}$ has full rank. Let $1\leq r \leq \dim C_1$ and let $1\leq \ell \leq s$ be such that $r\leq \dim C_\ell$. Then
$$
d_r(C)\leq d_r(C_\ell) \delta_\ell. 
$$
\end{prop}
\begin{proof}
Let $1\leq \ell\leq s$ be such that $r\leq \dim C_\ell$. We will obtain a subcode $D\subset C$ with $\dim D=r$ and $\abs{\supp(D)}=d_r(C_\ell) \delta_\ell$. First, we consider a subcode $D_\ell\subset C_\ell$ with $\dim D_\ell=r$ and $\abs{\supp(D_\ell)}=d_r(C_\ell)$. Let $f=\sum_{j=1}^\ell\lambda_j R_j$, with $\lambda_j\in \F_q$, be a codeword of $C_{R_\ell}$ with $\wt(f)=\delta_\ell$. Then we claim that
$$
D:=\{[\lambda_1 v_1,\dots,\lambda_\ell v_\ell,v_{\ell+1},\dots,v_s]\cdot A: v_1=v_2=\cdots =v_\ell \in D_\ell, v_{\ell+1}=v_{\ell+2}=\cdots=v_s=0 \}
$$
is a subcode of $C$ with $\dim D=r$ and $\abs{\supp(D)}= d_r(C_\ell)\cdot \delta_\ell$. It is clear that $D\subset C$ because $D_\ell\subset C_\ell\subset\cdots \subset C_1$, and $\dim D=r$ since $A$ has full rank. Let $v\in D_\ell$, then
$$
[\lambda_1 v,\dots,\lambda_\ell v]\cdot A_{\ell}=\left(\sum_{j=1}^\ell a_{j1}\lambda_j v,\dots, \sum_{j=1}^\ell a_{jh}\lambda_j v \right)=(vf_1,\dots,vf_h),
$$
where $f=(f_1,\dots,f_h)\in \F_q^h$, that is, $f_i$ is the $i$-th coordinate of $f$, for $1\leq i \leq h$. Hence,
$$
D=\{(vf_1,\dots,vf_h)\in C: v\in D_\ell \}.
$$
From this expression and the fact that $\abs{\supp(D_\ell)}=d_r(C_\ell)$, we obtain
$$
\abs{\supp_j(D)}=
\begin{cases}
   d_r(C_\ell) & \text{ if } f_j\neq 0,\\
   0 & \text{ if }f_j=0.
\end{cases}
$$
Since $\wt(f)=\delta_\ell$, we have $\abs{\supp(D)}=d_r(C_\ell)\cdot \delta_\ell$. 
\end{proof}

\begin{rem}
In the previous result, if $A$ is NSC, then by \cite[Prop. 7.2]{blackmoreMPC} we have $\delta_\ell=(h-\ell+1)$, for $1\leq \ell \leq s$. Moreover, if $A$ is triangular (that is, a column permutation of an upper triangular matrix), then the previous result holds even if the codes are not nested (this was already known to be true for the minimum distance \cite[Thm. 3.7]{blackmoreMPC}). Indeed, we just need to consider
$$
D':=\{[v_1,\dots,v_s]\cdot A: v_\ell \in D_\ell,\; v_{j}=0 \text{ if } j\neq \ell\},
$$
where we take $D_\ell$ as in the proof of Proposition \ref{P:upperbound}. 
Since $A$ is triangular, we have
$$
D'=\{ (a_{\ell1}v,\dots,a_{\ell h}v): v\in D_\ell\},
$$
where $a_{\ell j}$ is nonzero for exactly $h-\ell+1$ values of $j$, which implies $\abs{\supp(D)}=d_r(C_\ell)\cdot \delta_\ell$.
\end{rem}

Note that the previous result does not provide any upper bound if $r>\dim C_1$, and, when $r=\dim C_1$, it only gives $d_r(C)\leq h\cdot n=N$, which cannot be sharp if $\dim C_2\geq 1$ due to the monotony of the GHWs. This contrasts with the case of the minimum distance ($r=1$), where one gets that the minimum of the bounds provided in Proposition \ref{P:upperbound} is always sharp \cite[Thm. 1]{decodingMPC}. Nevertheless, for lower values of $r$, this bound performs well, as we see in the following example (and as we will see in Theorem \ref{T:GHWsRS}). 

\begin{ex}
Using the setting from Example \ref{ex:h3s2}, from Proposition \ref{P:upperbound}, we obtain
$$
d_2(C)\leq 3d_2(\RS(k_1))=9.
$$
Thus, from this we can also deduce that the bound given in Example \ref{ex:h3s2} is sharp. 
\end{ex}

\section{Examples for particular families of codes}\label{S:examples}

We start by considering Reed-Solomon codes $\RS(k)$ with dimension $k$ and length $n\leq q$, for which we know the GHWs from Remark \ref{r:ghwMDS}. In what follows, we denote 
\begin{equation}\label{eq:drRS}
d_r(\RS(k))=\begin{cases}
    0 & \text{ if } r=0 \\
    n-k+r & \text{ if } 1\leq r \leq k, \\ 
    \infty & \text{ if }k< r.
\end{cases}
\end{equation}

\begin{thm}\label{T:GHWsRS}
Let $1\leq k_2\leq k_1\leq n \leq q$, let $A$ be $2\times 2$ NSC matrix over $\fq$, and let $\RS(k_1,k_2):=[\RS(k_1),\RS(k_2)]\cdot A$. For $1\leq r\leq \dim \RS(k_1,k_2)=k_1+k_2$, we have
$$
d_r(\RS(k_1,k_2))= 
\begin{cases}
    2n+r-(k_1+k_2) & \text{ if } r>\max\{k_1-k_2,k_2\}, \\
    \min \{2 d_r(\RS(k_1)),d_r(\RS(k_2))\} & \text{ if } r\leq \max\{k_1-k_2,k_2\}.\\
\end{cases}
$$
\end{thm}
\begin{proof}
Let $\alpha_i\neq 0$, $\alpha_i\neq r$, for $1\leq i \leq 2$. First, we give a lower bound for $d_r(\RS(k_1,k_2))$ using Corollary \ref{C:h2}. By (\ref{eq:drRS}) we have
$$
B_{\alpha_1,\alpha_2}=\sum_{i=1}^2\max \{ n-k_1+r-\alpha_i,n-k_2+\alpha_{i+1}\},
$$
where $i+1$ is understood to be $i+1\bmod 2$. This can be expressed as
\begin{equation}\label{Balfa}
B_{\alpha_1,\alpha_2}=\begin{cases}
    2(n-k_1+r)-(\alpha_1+\alpha_2) & \text{ if } r\geq k_1-k_2+\alpha_1+\alpha_2,\\
    2(n-k_2)+\alpha_1+\alpha_2 & \text{ if } r< k_1-k_2+\alpha_1+\alpha_2.\\
\end{cases}
\end{equation}
We now study the minimum of $B_{\alpha_1,\alpha_2}$ for all $(\alpha_1,\alpha_2)\in Y_2$, with $\alpha_i\neq 0$, $\alpha_i\neq r$, using this expression. Recall that
$$
Y_2 = \left\{(\alpha_1,\alpha_2)\in [r]^2: \begin{array}{c}
        r-\dim C_1 \leq  \alpha_i \leq\dim C_2, \; 1\leq i\leq 2 \\
        \alpha_1+\alpha_2\leq r
        \end{array}
        \right\}.
$$
Let $\xi:=r-(k_1-k_2)$, and $z=\alpha_1+\alpha_2$. Consider $(\alpha_1,\alpha_2)\in Y_2$ with $\alpha_i\neq 0$, $\alpha_i\neq r$. Then we can rewrite (\ref{Balfa}) as 
\begin{equation*}\label{Bzeta}
B(z):=B_{\alpha_1,\alpha_2}=\begin{cases}
    2(n-k_2)+z & \text{ if } z> \xi, \\
    2(n-k_1+r)-z & \text{ if } z\leq \xi.\\
\end{cases}
\end{equation*}
As a function of $z$, we see that $B(z)$ is an increasing function for $z>\xi$ and a decreasing function for $z\leq \xi$. Thus, the minimum for $(\alpha_1,\alpha_2)\in Y_2$, $\alpha_i\neq 0$, $\alpha_i\neq r$, is always greater than or equal to 
$$
B(\xi)=2n+r-(k_1+k_2).
$$

Now we study the minimum of $B_{\alpha_1,\alpha_2}$ for $(\alpha_1,\alpha_2)\in Y_2$,  $\alpha_1=0$, $0<\alpha_2<r$. As before, we can write
$$
B_{0,\alpha_2}=\begin{cases}
    2(n-k_1+r)-\alpha_2 & \text{ if } r\geq k_1-k_2+\alpha_2,\\
    2n+r-(k_1+k_2) & \text{ if } r< k_1-k_2+\alpha_2.\\
    \end{cases}
$$
As a function of $\alpha_2$, this is constant for $\alpha_2>r-(k_1-k_2)$, and it is decreasing for $\alpha_2\leq r-(k_1-k_2)$. The minimum over $\alpha_2$, with $0<\alpha_2<r$, is greater than or equal to
$$
B_{0,r-(k_1-k_2)}=2n+r-(k_1+k_2)=B(\xi).
$$
The only cases left to check are $(\alpha_1,\alpha_2)=(0,0)$ and $(\alpha_1,\alpha_2)=(0,r)$, if they are in $Y_2$ (the rest of the cases are also covered by symmetry between $\alpha_1$ and $\alpha_2$). We have
$$
B_{0,0}=2(n-k_1+r)=2d_r(C_1),\; B_{0,r}=n-k_2+r=d_r(C_2). 
$$
Note that $(0,0)\in Y_2$ if and only if $r-k_1\leq 0$, and $(0,r)\in Y_2$ if and only if $r-k_1\leq 0$ and $r\leq k_2$ (this last condition implies $r\leq k_1$). It is straightforward to check that $B(\xi)\leq B_{0,0}$ if and only if $r\geq k_1-k_2$, $B_{0,r}\leq B(\xi)$ always (but $(0,r)\in Y_2$ only if $r\leq k_2$), and $B_{0,0}\leq B_{0,r}$ if and only if $r\leq k_1-k_2-(n-k_1)$. Therefore, by Corollary \ref{C:h2}, we obtain
\begin{equation}\label{eq:lowerRS}
d_r(\RS(k_1,k_2))\geq 
\begin{cases}
    B(\xi) & \text{ if } r>\max\{k_1-k_2,k_2\}, \\
    d_r(C_2) & \text{ if } k_1-k_2\leq k_2 \text{ and } k_1-k_2 < r \leq k_2,\\
    2d_r(C_1) & \text{ if } k_1-k_2 > k_2 \text{ and } k_2< r \leq k_1-k_2,\\
    d_r(C_2) & \text{ if } k_1-k_2-(n-k_1)< r \leq \min \{k_1-k_2,k_2\},\\
    2d_r(C_1) & \text{ if } r\leq k_1-k_2-(n-k_1). \\
\end{cases}
\end{equation}
It is straightforward to check that this lower bound is equal to the formula in the statement of the result (with the notation from (\ref{eq:drRS})). By Proposition \ref{P:upperbound} and Corollary \ref{C:singletongeneralizada}, the previous bound is sharp for $1\leq r \leq \dim \RS(k_1,k_2)$. 
\end{proof}

\begin{rem}
Note that the previous result shows that $\RS(k_1,k_2)$ is $t$-MDS, for $t=\max\{k_1-k_2,k_2\}$. Also note that the proof of Theorem \ref{T:GHWsRS} also works for any pair of MDS codes $C_1,C_2$ with dimensions $\dim C_1=k_1,\dim C_2=k_2$, such that $C_2\subset C_1$. 
\end{rem}

We turn our attention now to the family of Reed-Muller codes, which is closely related to MPCs, as we see next. We denote by $\RM_q(\nu,m)$ the Reed-Muller code of degree $\nu$ in $m$ variables over $\F_q$. We take $\F_q=\{\alpha_1,\dots,\alpha_q\}$. Let 
$$
\binom{\alpha_j}{\alpha_i}:=\frac{(\alpha_j-\alpha_1)\cdots (\alpha_j-\alpha_{i-1})}{(\alpha_i-\alpha_1)\cdots (\alpha_i-\alpha_{i-1})}, 
$$
where we understand that if $i=1$ or $i=j$ then $\binom{\alpha_j}{\alpha_i}=1$, and $\binom{\alpha_j}{\alpha_i}=0$ if and only if $1\leq j \leq i-1$. We consider the matrix
$$
\GRM_q:=\begin{pmatrix}
    \binom{\alpha_1}{\alpha_1} & \binom{\alpha_2}{\alpha_1} &\cdots & \binom{\alpha_q}{\alpha_1} \\
    \binom{\alpha_1}{\alpha_2} & \binom{\alpha_2}{\alpha_2} &\cdots & \binom{\alpha_q}{\alpha_2} \\
    \vdots & \vdots & \ddots & \vdots \\
    \binom{\alpha_1}{\alpha_q} & \binom{\alpha_2}{\alpha_q} &\cdots & \binom{\alpha_q}{\alpha_q} \\
\end{pmatrix}.
$$
In \cite[Section 5]{blackmoreMPC}, the authors prove that $\GRM_q$ is NSC, and they also prove the following result.

\begin{thm}
The Reed-Muller codes can be recursively defined by 
$$
\RM_q(\nu,0)=\begin{cases}
    \{0 \} & \text{ if } r<0, \\
    \F_q & \text{ if } r\geq 0,
\end{cases}
$$
and for $m\geq 1$
\begin{equation}\label{eq:recursiveRM}
\RM_q(\nu,m)=[\RM_q(\nu,m-1),\cdots,\RM_q(\nu-q+1,m-1)]\cdot \GRM_q. 
\end{equation}
\end{thm}
For $q=2$ and $q=3$, we get
$$
\GRM_2=\begin{pmatrix}
    1&1\\
    0&1
\end{pmatrix}, \; 
\GRM_3=\begin{pmatrix}
    1&1&1\\
    0&1&2\\
    0&0&1
\end{pmatrix}.
$$
In particular, this recovers the well-known result that binary Reed-Muller codes can be constructed recursively using the $(u,u+v)$ construction.

Another important aspect of Reed-Muller codes in this context is that their GHWs are known \cite{pellikaanGHWRM}. Therefore, they provide a family in which to test our bounds, in particular Corollary \ref{C:h2} and Theorem \ref{T:h3}. For example, for $q=2$, we can bound the GHWs of $\RM_2(\nu,m)$ with Corollary \ref{C:h2} using the GHWs of $\RM_2(\nu,m-1)$ and $\RM_2(\nu-1,m-1)$, and we can check if the bound is sharp because we know the true values of the GHWs of $\RM_2(\nu,m)$. We can proceed similarly for the case of $q=3$ using Theorem \ref{T:h3}. Note that we can apply our results since $\GRM_q$ is NSC and $\RM_q(\nu_1,m)\subset \RM_q(\nu_2,m)$ if $\nu_1\leq \nu_2$, i.e., the codes in (\ref{eq:recursiveRM}) are nested. 

For example, for $q=2$, we have
$$
\RM_2(\nu,m)=\{ (u,u+v): u\in \RM_2(\nu,m-1), v\in \RM_2(\nu-1,m-1)\}.
$$
For $1\leq r \leq \dim \RM_2(\nu,m)$, the bound from Corollary \ref{C:h2} with $C_1=\RM_2(\nu,m-1)$ and $C_2=\RM_2(\nu-1,m-1)$ would be 
\begin{equation}\label{eq:boundRM}
d_r(\RM_2(\nu,m))\geq  \min_{(\alpha_1,\alpha_2)\in Y_2}B_{\alpha_1,\alpha_2},
\end{equation}
where 
$$
\begin{aligned}
B_{\alpha_1,\alpha_2}=&\max \{d_r-\alpha_1(\RM_2(\nu,m-1)),d_{\alpha_2}(\RM_2(\nu-1,m-1))\\
&+\max \{d_r-\alpha_2(\RM_2(\nu,m-1)),d_{\alpha_1}(\RM_2(\nu-1,m-1)),
\end{aligned}
$$
and
$$
Y_2 = \left\{(\alpha_1,\alpha_2)\in [r]^2: \begin{array}{c}
        r-\dim \RM_2(\nu,m-1) \leq  \alpha_1 \leq\dim \RM_2(\nu-1,m-1), \\
        r-\dim \RM_2(\nu,m-1) \leq  \alpha_2 \leq\dim \RM_2(\nu-1,m-1),\\
        \alpha_1+\alpha_2\leq r
        \end{array}
        \right\}.
$$
Since $d_r(\RM_2(\nu,m))$ is known from \cite{pellikaanGHWRM}, we can compute the bound from (\ref{eq:boundRM}) and check if whether it gives the true minimum distance or not. We have done this for any $2\leq m\leq 10$ and any degree $0\leq \nu \leq m(q-1)$, and the bound (\ref{eq:boundRM}) coincides with the corresponding GHW in all of those cases. This not only seems to indicate that the bound from Corollary \ref{C:h2} might be sharp for this family, but also showcases the fact that it can be computed efficiently even for large codes. 

For the case $q=3$, we have computed the bound from Theorem \ref{T:h3} for $2\leq m\leq 3$ variables, which also gives the true value of the corresponding GHW of $\RM_3(\nu,m)$, $1\leq \nu \leq m(q-1)$. Since this bound is more computationally intensive to compute than the one from Corollary \ref{T:h3}, is not feasible to compute it for every possible degree for a larger number of variables. Notwithstanding the foregoing, we have tested a wide range of degrees for $4$ and $5$ variables, and we did not find any case in which the bound did not coincide with the GHW.

\section*{Declarations}
\subsection*{Conflict of interest} The author declares that he has no conflict of interest.


\begin{thebibliography}{10}

\bibitem{munueraGHWhermitica}
A.~I. Barbero and C.~Munuera.
\newblock The weight hierarchy of {H}ermitian codes.
\newblock {\em SIAM J. Discrete Math.}, 13(1):79--104, 2000.

\bibitem{beelenGHWcartesian}
P.~Beelen and M.~Datta.
\newblock Generalized {H}amming weights of affine {C}artesian codes.
\newblock {\em Finite Fields Appl.}, 51:130--145, 2018.

\bibitem{blackmoreMPC}
T.~Blackmore and G.~H. Norton.
\newblock Matrix-product codes over {$\Bbb F_q$}.
\newblock {\em Appl. Algebra Engrg. Comm. Comput.}, 12(6):477--500, 2001.

\bibitem{eduardoGHWHyperbolic}
E.~Camps-Moreno, I.~Garc\'{\i}a-Marco, H.~H. L\'{o}pez, I.~M\'{a}rquez-Corbella, E.~Mart\'{\i}nez-Moro, and E.~Sarmiento.
\newblock On the generalized {H}amming weights of hyperbolic codes.
\newblock {\em Journal of Algebra and Its Applications}, 23(07):2550062, 2024.

\bibitem{sanjoseGHWNT}
E.~Camps-Moreno, H.~H. López, G.~L. Matthews, and R.~San-Jos\'e.
\newblock The weight hierarchy of decreasing norm-trace codes.
\newblock {\em Designs, Codes and Cryptography, to appear. ArXiv 2411.13375}, 2024.

\bibitem{liuMPChomogeneous1}
Y.~Fan, S.~Ling, and H.~Liu.
\newblock Homogeneous weights of matrix product codes over finite principal ideal rings.
\newblock {\em Finite Fields Appl.}, 29:247--267, 2014.

\bibitem{fanMPC}
Y.~Fan, S.~Ling, and H.~Liu.
\newblock Matrix product codes over finite commutative {F}robenius rings.
\newblock {\em Des. Codes Cryptogr.}, 71(2):201--227, 2014.

\bibitem{fengGHWsCyclic}
G.~L. Feng, K.~K. Tzeng, and V.~K. Wei.
\newblock On the generalized {H}amming weights of several classes of cyclic codes.
\newblock {\em IEEE Trans. Inform. Theory}, 38(3):1125--1130, 1992.

\bibitem{galindoMPCLRC}
C.~Galindo, F.~Hernando, C.~Munuera, and D.~Ruano.
\newblock Locally recoverable codes from the matrix-product construction.
\newblock {\em ArXiv 2310.15703}, 2023.

\bibitem{galindoQuantumMPC}
C.~Galindo, F.~Hernando, and D.~Ruano.
\newblock New quantum codes from evaluation and matrix-product codes.
\newblock {\em Finite Fields Appl.}, 36:98--120, 2015.

\bibitem{guruswammiGHWlistdecodingTensorInterleaved}
P.~Gopalan, V.~Guruswami, and P.~Raghavendra.
\newblock List decoding tensor products and interleaved codes.
\newblock {\em SIAM J. Comput.}, 40(5):1432--1462, 2011.

\bibitem{guruswammiGHWlistdecoding}
V.~Guruswami.
\newblock List decoding from erasures: bounds and code constructions.
\newblock {\em IEEE Trans. Inform. Theory}, 49(11):2826--2833, 2003.

\bibitem{pellikaanGHWRM}
P.~Heijnen and R.~Pellikaan.
\newblock Generalized {H}amming weights of {$q$}-ary {R}eed-{M}uller codes.
\newblock {\em IEEE Trans. Inform. Theory}, 44(1):181--196, 1998.

\bibitem{hellesethGHWLinearcodes}
T.~Helleseth, T.~Kl\o~ve, and O.~y. Ytrehus.
\newblock Generalized {H}amming weights of linear codes.
\newblock {\em IEEE Trans. Inform. Theory}, 38(3):1133--1140, 1992.

\bibitem{hellesethGHWcyclic}
T.~Helleseth, T.~Kl{\o}ve, and J.~Mykkeltveit.
\newblock The weight distribution of irreducible cyclic codes with block length {$n\sb{1}((q\sp{l}-1)/N)$}.
\newblock {\em Discrete Math.}, 18(2):179--211, 1977.

\bibitem{hernandoListDecodingMPC}
F.~Hernando, T.~H{\o}holdt, and D.~Ruano.
\newblock List decoding of matrix-product codes from nested codes: an application to quasi-cyclic codes.
\newblock {\em Adv. Math. Commun.}, 6(3):259--272, 2012.

\bibitem{decodingMPC}
F.~Hernando, K.~Lally, and D.~Ruano.
\newblock Construction and decoding of matrix-product codes from nested codes.
\newblock {\em Appl. Algebra Engrg. Comm. Comput.}, 20(5-6):497--507, 2009.

\bibitem{hernandoDecodingMPC2}
F.~Hernando and D.~Ruano.
\newblock Decoding of matrix-product codes.
\newblock {\em J. Algebra Appl.}, 12(4):1250185, 15, 2013.

\bibitem{janwaGHWCyclic}
H.~Janwa and A.~K. Lal.
\newblock On the generalized {H}amming weights of cyclic codes.
\newblock {\em IEEE Trans. Inform. Theory}, 43(1):299--308, 1997.

\bibitem{jitmanSelforthogonalMPCHerm}
S.~Jitman and T.~Mankean.
\newblock Matrix-product constructions for {H}ermitian self-orthogonal codes.
\newblock {\em Chamchuri J. Math.}, 9:35--51, 2017.

\bibitem{matsumotoRGRW}
J.~Kurihara, R.~Matsumoto, and T.~Uyematsu.
\newblock Relative generalized rank weight of linear codes and its applications to network coding.
\newblock {\em IEEE Trans. Inform. Theory}, 61(7):3912--3936, 2015.

\bibitem{matsumotoRGHW}
J.~Kurihara, T.~Uyematsu, and R.~Matsumoto.
\newblock Secret sharing schemes based on linear codes can be precisely characterized by the relative generalized hamming weight.
\newblock {\em IEICE Trans. Fundam. Electron. Commun. Comput. Sci.}, E95.A(11):2067--2075, 2012.

\bibitem{liuMPChomogeneous2}
H.~Liu and J.~Liu.
\newblock Homogeneous metric and matrix product codes over finite commutative principal ideal rings.
\newblock {\em Finite Fields Appl.}, 64:101666, 29, 2020.

\bibitem{luoMPCLRC}
G.~Luo, M.~F. Ezerman, and S.~Ling.
\newblock Three new constructions of optimal locally repairable codes from matrix-product codes.
\newblock {\em IEEE Trans. Inform. Theory}, 69(1):75--85, 2023.

\bibitem{LuoSymbolpairMPC}
G.~Luo, M.~F. Ezerman, S.~Ling, and X.~Pan.
\newblock New families of {MDS} symbol-pair codes from matrix-product codes.
\newblock {\em IEEE Trans. Inform. Theory}, 69(3):1567--1587, 2023.

\bibitem{jitmanSelforthogonalMPCEU}
T.~Mankean and S.~Jitman.
\newblock Matrix-product constructions for self-orthogonal linear codes.
\newblock In {\em 2016 12th International Conference on Mathematics, Statistics, and Their Applications (ICMSA)}, pages 6--10, 2016.

\bibitem{umbertoGHWandGRW}
U.~Mart\'inez-Pe\~nas.
\newblock On the similarities between generalized rank and {H}amming weights and their applications to network coding.
\newblock {\em IEEE Trans. Inform. Theory}, 62(7):4081--4095, 2016.

\bibitem{munueraGHWGoppa}
C.~Munuera.
\newblock On the generalized {H}amming weights of geometric {G}oppa codes.
\newblock {\em IEEE Trans. Inform. Theory}, 40(6):2092--2099, 1994.

\bibitem{existenceGRW}
F.~Oggier and A.~Sboui.
\newblock On the existence of generalized rank weights.
\newblock In {\em 2012 International Symposium on Information Theory and its Applications}, pages 406--410, 2012.

\bibitem{ferruhMPC}
F.~\"{O}zbudak and H.~Stichtenoth.
\newblock Note on {N}iederreiter-{X}ing's propagation rule for linear codes.
\newblock {\em Appl. Algebra Engrg. Comm. Comput.}, 13(1):53--56, 2002.

\bibitem{pellikaanlibro}
R.~Pellikaan, X.-W. Wu, S.~Bulygin, and R.~Jurrius.
\newblock {\em Codes, cryptology and curves with computer algebra}.
\newblock Cambridge University Press, Cambridge, 2018.

\bibitem{sanjoseRecursivePRM}
R.~San-Jos\'e.
\newblock A recursive construction for projective {R}eed-{M}uller codes.
\newblock {\em IEEE Trans. Inform. Theory}, 70(12):8511--8523, 2024.

\bibitem{aschMPC}
B.~van Asch.
\newblock Matrix-product codes over finite chain rings.
\newblock {\em Appl. Algebra Engrg. Comm. Comput.}, 19(1):39--49, 2008.

\bibitem{weiGHW}
V.~K. Wei.
\newblock Generalized {H}amming weights for linear codes.
\newblock {\em IEEE Trans. Inform. Theory}, 37(5):1412--1418, 1991.

\bibitem{yangGHWCyclic}
M.~Yang, J.~Li, K.~Feng, and D.~Lin.
\newblock Generalized {H}amming weights of irreducible cyclic codes.
\newblock {\em IEEE Trans. Inform. Theory}, 61(9):4905--4913, 2015.

\end{thebibliography}

\end{document}